\documentclass[11pt,a4paper]{article}
\usepackage{fullpage}
\usepackage{tikz}
\newcommand*\circled[1]{\tikz[baseline=(char.base)]{
            \node[shape=circle,draw,inner sep=2pt] (char) {#1};}}

\usepackage{amssymb,amsmath,amsfonts,amstext,amsthm}
\usepackage{bbm}
\usepackage{url}
\usepackage{graphicx}
\usepackage{subcaption}
\usepackage{verbatim}
\usepackage{multirow}
\usepackage[ruled,vlined]{algorithm2e}

\usepackage{times}

\newtheorem{theorem}{Theorem}
\newtheorem{lemma}[theorem]{Lemma}

\theoremstyle{definition}
\newtheorem{example}{Example}

\newcommand{\pp}{\mathbf{p}}
\newcommand{\s}{\mathbf{s}}
\newcommand{\qq}{\mathbf{q}}
\newcommand{\xx}{\mathbf{x}}
\newcommand{\UU}{\mathcal{U}}
\newcommand{\CC}{\mathcal{C}}
\newcommand{\SW}{\text{SW}}
\newcommand{\R}{\mathbb{R}}
\newcommand{\yy}{\mathbf{y}}
\newcommand{\bigO}{{\cal O}}

\newcommand{\one}[1]{\mathbbm{1}\{#1\}}

\title{{\bf Efficiency and complexity of price competition among single-product vendors}\thanks{A preliminary version of this paper appeared in {\em Proceedings of the 24th International Joint Conference on Artificial Intelligence (IJCAI)}, pages 25--31, 2015. This work was partially supported by Caratheodory research grant E.114 from the University of Patras and by a PhD scholarship from the Onassis Foundation.}}

\author{%
	\makebox[.25\linewidth]{Ioannis Caragiannis}
	\and \makebox[.25\linewidth]{Xenophon Chatzigeorgiou}
	\and \makebox[.25\linewidth]{Panagiotis Kanellopoulos}
  	\and \makebox[.25\linewidth]{George A. Krimpas}
  	\and \makebox[.25\linewidth]{Nikos Protopapas}
  	\and \makebox[.25\linewidth]{Alexandros A. Voudouris}
}
\date{University of Patras \& CTI ``Diophantus''}

\begin{document}
\maketitle

\begin{abstract}
Motivated by recent progress on pricing in the AI literature, we study marketplaces that contain multiple vendors offering identical or similar products and unit-demand buyers with different valuations on these vendors. The objective of each vendor is to set the price of its product to a fixed value so that its profit is maximized. The profit depends on the vendor's price itself and the total volume of buyers that find the particular price more attractive than the price of the vendor's competitors. We model the behaviour of buyers and vendors as a two-stage full-information game and study a series of questions related to the existence, efficiency (price of anarchy) and computational complexity of equilibria in this game. To overcome situations where equilibria do not exist or exist but are highly inefficient, we consider the scenario where some of the vendors are subsidized in order to keep prices low and buyers highly satisfied.
\end{abstract}

\section{Introduction}
We focus on marketplaces that contain multiple vendors offering a single product and unit-demand buyers. For example, we may think of software development companies, each offering an operating system. Each potential user is interested in buying one operating system from some software company and has preferences over the different options available in the market. Her final choice depends not only on her preferences but also on the prices of the available products; eventually, each user will choose the product with the best value for money, or will simply abstain from purchasing a product if the available options are not satisfactory for her. In turn, vendors are aware of this buyer behaviour and aim to set the price of their product to a value that will maximize their profit. In particular, the dilemma a vendor faces is to select between a very small price that will guarantee a large market share or a huge price that will be attractive only to a few buyers. Of course, there are usually many options in between, and coming up with a pricing that will maximize profits in such an environment is challenging.

We model the above scenario as a two-stage full-information game (with both the vendors and the buyers as players) which we call a {\em price competition game}. In the first stage, each vendor selects the price of its product among a set of viable price values (i.e., the price values that are above a fixed production cost per unit of product). Buyers have unit demands and (possibly different) valuations for vendors. Together with the valuations of buyers, a vector of prices (with one price per vendor) determines in a second stage the most attractive vendor for each buyer. Each vendor has full information about the valuations of buyers and can predict their behaviour. The objective of each vendor is to set its price so that its profit (i.e., volume of buyers it attracts times the difference of price and production cost) is maximized given the prices of the other vendors.

\paragraph{A high-level overview of our contribution} We present a list of results for these price competition games. Our starting point is the observation that equilibria (i.e., buyers-to-vendors assignments and corresponding prices so that all vendors and all buyers are satisfied) are guaranteed to exist only when all buyers have the same valuations; price competition games with buyers belonging to at least two different types (with respect to their valuations) may not have equilibria. Even when equilibria exist, they may be highly suboptimal. We use the notion of the price of anarchy (introduced by Koutsoupias and Papadimitriou~\cite{KP09}; see also Papadimitriou~\cite{P01}) to quantify how low the social welfare of equilibria can be compared to the optimal one. The social welfare is essentially the sum of buyer utilities and vendor profits. We also formulate several variations of equilibrium computation problems and present complexity results about them. These results range from polynomial-time algorithms (e.g., for the problem of determining prices that form an equilibrium together with a given buyers-to-vendors assignment) to hardness results (e.g., for the general problem of deciding whether a given price competition game admits an equilibrium). Motivated by the negative results on the existence and quality of equilibria, we investigate whether efficient buyers-to-vendors assignments can be enforced as equilibria by {\em subsidizing} the vendors. Our main contribution here is conceptual: subsidies can indeed overcome the drawbacks of price competition. Our technical contributions include tight bounds on the amount of subsidies sufficient to enforce a social welfare-maximizing buyers-to-vendors assignment as an equilibrium, as well as inapproximability results for the problem of minimizing the amount of subsidies sufficient to do so.

\paragraph{Related work} AI literature has long considered settings where vendors and buyers interact in electronic marketplaces (see e.g.~\cite{SK05}). Our model is very similar to (and actually inspired from) the one considered by Meir et al.~\cite{MLTB14} who focus on the impact of discounts (i.e., prices that are decreasing functions of demand) on vendors' profit compared to fixed prices. After observing that price discounts have no impact at all in the full-information setting, they mostly focus on a Bayesian setting with uncertainty on buyers' valuations. In contrast, we restrict our attention to the full-information model and consider only fixed prices. As we will see, this simple setting is very rich from the computational point of view. With the work of Meir et al.~\cite{MLTB14} as an exception, our assumptions differ significantly from most of the literature on price competition. For example, unlike early models such as the ones proposed by Cournot and Bertrand (see the book of Mas-Colell et al.~\cite{MWG95}) as well as very recent refinements (e.g., the work of Babaioff et al.~\cite{BLN13}), we assume that all vendors have unlimited supply. Also, contrary to other recent models that consider buyers with combinatorial valuations for bundles of different products as in the papers of Guruswami et al.~\cite{GHK+05}, Chawla and Roughgarden~\cite{CR08}, Babaioff et al.~\cite{BNPL14}, Lev et al.~\cite{LOBR15} and more, we specifically assume that each buyer is interested in obtaining just a single product. In this way, the decision each buyer faces is rather trivial and this allows us to concentrate on the competition between the vendors. The competition of single-product vendors with unlimited supply differentiates our model also from the concept of Walrasian equilibrium, where the emphasis is on assigning prices to different items so that all buyers are satisfied, see e.g., \cite{GS99,GHK+05}. In addition, unlike variations of the models studied in these papers (such as in \cite{KC82}), vendors do not use price discrimination in our case. On a more technical level, we implicitly assume an infinite number of buyers and use the notion of buyer types to distinguish between sets of buyers; this is a less important difference of our model to previous work on pricing. Very recent work, following the conference version of this manuscript, includes Anshelevich and Sekar \cite{AS15} that introduces seller costs to Bertrand competition and allows for arbitrary supply. They study a two-stage full-information pricing game, where vendors sell edges on a directed graph and buyers wish to buy paths that connect a source to a sink node, and characterize the structure and properties of equilibria of that game. In addition, Borodin et al. \cite{BLS16} consider a pricing game with a single budget-constrained buyer and strategic vendors and characterize equilibria for a large class of valuation functions.

The use of subsidies in price competition games suggests yet another way of introducing external monetary incentives in games; such incentives (or disincentives) have been considered in many different contexts. Much of the work in {\em mechanism design} uses such incentives to motivate players to act truthfully (see \cite{N07} for an introduction to the field). The (apparently non-exhaustive) list also includes their use in {\em cooperative game theory} in order to encourage coalitions of players to reach stability~\cite{BEM+09} and as a means to stabilize normal form games \cite{MT04}. As in \cite{ACFK12} and \cite{BLNO10}, the use of monetary incentives in the current paper aims to improve efficiency. Monetary disincentives like {\em taxes} have been used to improve the efficiency of network routing (see \cite{CDR06} and the references therein for a relatively recent approach that extends early developments in the literature of the {\em economics of transportation}) and, in the recent AI literature, in boolean games \cite{WEKL13}.

\paragraph{Roadmap} The rest of the paper is structured as follows. We begin with preliminary definitions in Section \ref{sec:prelim}. Then, we consider questions about the existence of equilibria in price competition games and their price of anarchy in Section \ref{sec:quality}. We formulate computational problems for equilibria and study related complexity questions in Section \ref{sec:complexity}. We investigate the potential of subsidizing specific vendor prices in Section \ref{sec:subsidies} and, finally, we conclude with open questions in Section~\ref{sec:open}.

\section{Preliminaries}\label{sec:prelim}
Our setting includes a set $M$ containing $m$ vendors targeting a large population of buyers. The buyers are classified into $n$ buyer types from a set $N$. We denote by $\mu_i$ the volume of buyer type $i$; for example, $\mu_i=3$ can have the meaning that there are $3$ million buyers of type $i$. Each of these buyers has a non-negative valuation $v_{ij}$ for vendor $j$ (representing the satisfaction each buyer of type $i$ has when buying the product of vendor $j$). Essentially, the buyers are classified into types depending on the valuations they have for the vendors.

Each vendor $j$ has a non-negative cost $c_j$ per unit of product; we refer to $c_j$ as the {\em production cost} of vendor $j$. The objective of each vendor $j$ is to determine a price $p_j$ for its product; naturally, $p_j\geq c_j$ so that the vendor always has non-negative profit. A {\em price vector} $\pp=(p_1,...,p_m)$ (containing a price per vendor) defines a {\em demand set} $D_i(\pp)$ which, for each buyer type $i$, denotes the set of vendors that maximize the utility of the buyers of type $i$, i.e., $$D_i(\pp)=\arg\max_{j\in M}\{v_{ij} - p_j\}.$$

Intuitively, the demand set for buyers of type $i$ consists of the most attractive vendors for these buyers. We assume that the operator $\arg\max_{j\in M}$ returns (a set containing) an artificial vendor which represents an ``abstain'' option that a buyer has when its maximum utility (over all vendors) is non-positive. With a small abuse of notation, we introduce an extra vendor into $M$ in order to represent this abstain option for buyers; this vendor has production cost of $0$, it always has a price of $0$, and the valuations of buyers for it are $0$.

A {\em buyers-to-vendors assignment} (or, simply, an {\em assignment}) is represented by an $n\times (m+1)$ matrix $\xx$ which denotes how the volume of the buyers of each type is split among different vendors. In particular, the entry $x_{ij}$ denotes the volume of buyers of type $i$ that are assigned to vendor $j$ and it must be $\sum_{j\in M}{x_{ij}}=\mu_i$ for every buyer type $i$. An assignment $\xx$ is {\em consistent} to a price vector $\pp$ if $x_{ij}>0$ implies $j\in D_i(\pp)$. We can interpret such an assignment as maximizing the utility of buyers given the price vector $\pp$. We will denote by $$t_i(\xx,\pp)=\sum_{j\in M}{x_{ij}(v_{ij}-p_j)}$$ the total utility of buyers of type $i$ given a price vector $\pp$ and a consistent assignment $\xx$.

We study the game induced among vendors and buyers and use the term {\em price competition game} to refer to it. This can be thought of as a two stage game. At a first stage, the strategy of each vendor is its price. At a second stage, the buyers respond to these prices as described above. The utility of vendor $j$, when the vendors use a price vector $\pp$ and the buyers are assigned to vendors according to an assignment $\xx$ that is consistent to $\pp$, is defined as $$u_j(\xx,\pp)=(p_j-c_j)\sum_{i\in N}{x_{ij}}.$$ Vendors are utility-maximizers. A price vector $\pp$ and a consistent assignment $\xx$ form a {\em (pure Nash) equilibrium} when for every vendor $j$, the price $p_j$ maximizes the utility $u_j(\yy,(p'_j,\pp_{-j}))$ among all prices $p'_j\geq c_j$ and all assignments $\yy$ that are consistent to $(p'_j,\pp_{-j})$. Here, the notation $(p'_j,\pp_{-j})$ is used to represent the price vector where all vendors besides $j$ use the prices in $\pp$ and vendor $j$ has deviated to price $p'_j$.

The {\em social welfare} of an assignment $\xx$ is defined as $$\SW(\xx) = \sum_{i\in N}{\sum_{j\in M}{x_{ij}(v_{ij}-c_j)}}.$$
This definition does not require the assignment $\xx$ to be consistent to a price vector and can be used to define the {\em optimal social welfare} as
$$\SW^* = \sum_{i\in N}{\mu_i\max_{j\in M}\{v_{ij}-c_j\}}.$$
When the assignment $\xx$ is consistent to a price vector $\pp$, the social welfare can be equivalently seen as the total utility of vendors and buyers since
\begin{eqnarray*}
\SW(\xx) &=& \sum_{i\in N}{\sum_{j\in M}{x_{ij}(v_{ij}-c_j)}}\\
&=& \sum_{i\in N}{\sum_{j\in M}{x_{ij}(v_{ij}-p_j)}} + \sum_{j\in M}{(p_j-c_j) \sum_{i\in N}{x_{ij}}}\\
&=& \sum_{i\in N}{t_i(\xx,\pp)} + \sum_{j\in M}{u_j(\xx,\pp)}.
\end{eqnarray*}

The {\em price of anarchy} of a price competition game is the ratio of the optimal social welfare over the minimum social welfare among all equilibria. Of course, this is well-defined only for price competition games that do have equilibria.

\begin{example}
Consider a price competition game as shown in Table \ref{tab:example}.
\begin{table}[htbp]
\centerline{
\begin{tabular}{l|c c c}
  & vn $1$& vn $2$& vn $3$\\[-0.5em]
  & $c_1=0$& $c_2=1$& $c_3=2$\\\hline
  bt $A$: $\mu_A=1$ & $2$ & $1$ & $0$ \\
  bt $B$: $\mu_B=1/2$ & $1$ & $2$ & $1$ \\
  bt $C$: $\mu_C=2/3$ & $1$ & $0$ & $2$
\end{tabular}}
  \caption{An example with three buyer types and three vendors. The first row shows the vendors (vn) and their production cost and the leftmost column represents buyer types (bt) and their volume. Unlabeled cells contain valuations $v_{ij}$ for $i\in \{A,B,C\}$ and $j\in \{1,2,3\}$.}\label{tab:example}
  \end{table}

Note that the buyers-to-vendors assignment $\xx$ where $x_{A1} = \mu_A$, $x_{B2} = \mu_B$, $x_{C3} = \mu_C$, while all other entries in $\xx$ are equal to $0$, is consistent to the price vector $\pp = (2,2,2)$ and maximizes the social welfare. In particular, vendor  $1$ obtains utility $u_1(\xx,\pp)=\mu_A(p_1-c_1) = 2$, vendor $2$ obtains utility $u_2(\xx,\pp) = \mu_B(p_2-c_2) = 1/2$, while vendor $3$ has $u_3(\xx,\pp) = \mu_C(p_3-c_3) = 0$; since all buyers obtain utility equal to $0$, it holds that $\SW(\xx) = 5/2$. We remark that $(\xx,\pp)$ is not a pure Nash equilibrium, as vendor $1$ can lower its price to $p_1' = 1-\epsilon$, where $\epsilon>0$ is arbitrarily small, and attract all buyers; the obtained utility under the new assignment $\yy$ is $u_1(\yy,(p_1',\pp_{-1})) = 13/6\cdot(1-\epsilon)$, i.e., $u_1(\yy,(p_1',\pp_{-1})) > u_1(\xx,\pp)$.
\end{example}

In the following, we sometimes use the abbreviation $x^+$ instead of $\max\{0,x\}$ and write $[\ell]$ instead of the set $\{1, 2, ..., \ell\}$ for an integer $\ell\geq 1$.

\section{Existence and quality of equilibria}\label{sec:quality}
As a warm up, we present a negative result that reveals a strong relation of the price of anarchy of price competition games to the number of buyer types.
\begin{lemma}\label{lem:lower}
There are one-vendor price competition games with price of anarchy that is arbitrarily close to $n$.
\end{lemma}

\begin{proof}
Consider a price competition game with $n$ buyer types and one vendor with a production cost of $0$. Let $\alpha\in (0,1)$. The volume of a buyer type $i \in N$ is $\mu_i = \alpha^{i-1}$. The valuation of buyers of type $i$ is $v_i = \left(\sum_{j=i}^n{\mu_j}\right)^{-1}$ for $i\in [n-1]$ and $v_n=(1+\alpha)/\mu_n$. By setting its price up to $v_i$ for $i\in [n-1]$, the vendor can only get a utility of at most $v_i \sum_{j=i}^n{\mu_j}$ by attracting the buyers of type $i, i+1, ..., n$. By the definition of $v_i$, this utility is at most $1$. This is smaller than the utility the vendor would have by selecting a price of $v_n$ and attracting only the buyers of type $n$ (the remaining buyers simply abstain). This is an equilibrium in which the utility of the vendor (as well as the social welfare) is $1+\alpha$. In contrast, the social welfare of the assignment in which all buyers are assigned to the vendor is $\sum_{i=1}^n{\mu_i v_i} \geq (1-\alpha)n$; the inequality holds by the definition of $v_i$ and since
$$\sum_{j=i}^n{\mu_j} = \sum_{j=i}^n{\alpha^{j-1}} = \alpha^{i-1}\sum_{j=0}^{n-i}\alpha^j \leq \mu_i\sum_{j=0}^\infty{\alpha^j}=\mu_i(1-\alpha)^{-1}.$$
The price of anarchy is then at least $(1-\alpha)n/(1+\alpha)$ which can become arbitrarily close to $n$ by selecting $\alpha$ appropriately.
\end{proof}

Interestingly, the price of anarchy does not depend on any other quantity and the lower bound of Lemma \ref{lem:lower} is tight.
\begin{theorem}\label{thm:poa-upper}
The price of anarchy of any price competition game with $n$ buyer types is at most $n$.
\end{theorem}

\begin{proof}
Consider an equilibrium $(\xx,\pp)$ of a price competition game. We first claim that if buyers of some type $i$ are split between two vendors $j$ and $j'$, then it must be $p_j=c_j$ and $p_{j'}=c_{j'}$ (hence, the two vendors have zero utility) and the assignment in which all these buyers are assigned to vendor $j$ without changing the prices is still an equilibrium and has the same social welfare. This is due to the fact that, at equilibrium, the utilities of buyers of type $i$ assigned to $j$ and $j'$ should be the same. Hence, if one of the two vendors had a price strictly higher than its production cost, it could increase its utility by negligibly decreasing its price; this would result in attracting all buyers of type $i$ previously assigned to $j$ and $j'$. So, by moving all buyers of type $i$ from vendor $j'$ to vendor $j$, we still have an assignment that is consistent to $\pp$ in which the utilities of buyers and vendors do not change. Clearly, this new assignment is an equilibrium with a social welfare equal to the initial one.

So, without loss of generality, we consider an equilibrium $(\xx,\pp)$ such that, for every $i$, all buyers of type $i$ are assigned to the same vendor $j$, i.e., $x_{ij}=\mu_i$. We denote by $\eta(i)$ the vendor where the buyers of type $i$ are assigned in $\xx$. Also, we denote by $o(i)$ the vendor to which the buyers of type $i$ are assigned in an optimal assignment. Again, if there exists an optimal assignment where buyers of the same type $i$ are split among different vendors, we can move all these buyers to a single vendor without decreasing the social welfare, hence $o(i)$ is well defined.  We can further assume that when $\eta(i)\not=o(i)$, this implies that $v_{i,o(i)}-c_{o(i)}>v_{i,\eta(i)}-c_{\eta(i)}$. If this is not the case and it is $v_{i,o(i)}-c_{o(i)}=v_{i,\eta(i)}-c_{\eta(i)}$, we can consider the optimal assignment that assigns the buyers of type $i$ to vendor $\eta(i)$.

We will show that
\begin{eqnarray}\label{eq:player-i}
t_i(\xx,\pp) + u_{o(i)}(\xx,\pp) &\geq& \mu_i(v_{i,o(i)}-c_{o(i)})
\end{eqnarray}
for every $i\in N$. Then, the following derivation can prove the theorem:
\begin{eqnarray*}
n \cdot \SW(\xx) &\geq &\sum_i{t_i(\xx,\pp)} + n\cdot \sum_j{u_j(\xx,\pp)}\\
&\geq &\sum_i{\left(t_i(\xx,\pp)+u_{o(i)}(\xx,\pp) \right)}\\
&\geq & \sum_i{\mu_i \left( v_{i,o(i)}-c_{o(i)} \right)}\\
&=& \SW^*.
\end{eqnarray*}
The first inequality follows by the definition of the social welfare, the second one follows from the fact that the function $o(\cdot)$ can assign at most all $n$ buyer types to the same vendor, the third one follows from (\ref{eq:player-i}) and the last equality follows from the definition of the optimal social welfare.

It remains to prove inequality (\ref{eq:player-i}). If $\eta(i)=o(i)$, we use the fact that vendor $o(i)$ attracts (at least) the buyers of type $i$ at equilibrium. Hence,
\begin{eqnarray*}
t_i(\xx,\pp) + u_{o(i)}(\xx,\pp) &\geq& \mu_i(v_{i,o(i)}-p_{o(i)})+\mu_i(p_{o(i)}-c_{o(i)})\\
&= &\mu_i(v_{i,o(i)}-c_{o(i)}).
\end{eqnarray*}

If $\eta(i) \neq o(i)$, let $q_{o(i)}=v_{i,o(i)}-v_{i,\eta(i)}+p_{\eta(i)}$. First, we claim that $q_o(i)$ is a valid price for vendor $o(i)$. Indeed,
by our assumption $v_{i,o(i)}-c_{o(i)}>v_{i,\eta(i)}-c_{\eta(i)}$ above and since $p_{\eta(i)}\geq c_{\eta(i)}$, we have
\begin{eqnarray*}
q_{o(i)} &=& v_{i,o(i)} - v_{i,\eta(i)}+ p_{\eta(i)}\\
&>& c_{o(i)}-c_{\eta(i)}+p_{\eta(i)}\\
&\geq & c_{o(i)}.
\end{eqnarray*}
This means that vendor $o(i)$ can consider deviating to any price value $\delta$ from the non-empty interval $[c_{o(i)}, q_{o(i)})$. From the discussion above, it is $v_{i,o(i)}-\delta > v_{i,\eta(i)}-p_{\eta(i)}$, which means that, by deviating to $\delta$, vendor $o(i)$ can attract the buyers of type $i$ from vendor $\eta(i)$. Using the equilibrium condition (and denoting by $\xx'$ the resulting consistent assignment when vendor $o(i)$ deviates to price $\delta$), we have that
\begin{eqnarray*}
u_{o(i)}(\xx,\pp) &\geq &u_{o(i)}(\xx',(\delta,\pp_{-o(i)}))\\
&\geq& \mu_i(\delta-c_{o(i)}).
\end{eqnarray*}
Since the above inequality holds for any $\delta<q_{o(i)}$, it must also be
\begin{eqnarray*}
u_{o(i)}(\xx,\pp) &\geq & \mu_i(q_{o(i)}-c_{o(i)})\\
&=& \mu_i(v_{i,o(i)}-c_{o(i)})-\mu_i(v_{i,\eta(i)}-p_{\eta(i)})\\
&=& \mu_i(v_{i,o(i)}-c_{o(i)})-t_i(\xx,\pp).
\end{eqnarray*}
This completes the proof of the theorem.
\end{proof}

Recall that the upper bound on the price of anarchy is meaningful only for games that admit equilibria. In the following, we show that games with one buyer type always have equilibria (and, by Theorem \ref{thm:poa-upper}, all equilibria have optimal social welfare since the price of anarchy is $1$) while the existence of a second buyer type may result to instability.
\begin{lemma}\label{lem:existence}
Price competition games with one buyer type always have at least one equilibrium.
\end{lemma}

\begin{proof} Consider a price competition game with one buyer type. We use the simplified notation $v_j$ to denote the valuation of the buyers for vendor $j$. Let $j^*\in\arg\max_{j\in M}\{v_j-c_j\}$ and $j'\in\arg\max_{j\in M\setminus j^*}\{v_j-c_j\}$ be two vendors with the highest values for the difference $v_j-c_j$. Set $p_{j^*}=v_{j^*}-v_{j'}+c_{j'}$ and $p_j=c_j$ for any vendor $j\neq j^*$. We claim that this price vector together with the consistent assignment $\xx$ that assigns the buyers to vendor $j^*$ is an equilibrium. Indeed, no vendor $j \neq j^*$ has any incentive to change its price; a decrease would result in negative utility while an increase would not change the assignment. Moreover, vendor $j^*$ has no incentive to change its price; a decrease can only lower its utility while an increase would result in a new assignment in which all buyers are attracted by vendor $j'$.
\end{proof}

\begin{lemma}\label{lem:existence2}
There exists a price competition game with two buyer types that admits no equilibrium.
\end{lemma}

\begin{proof}
Consider the price competition game with two buyer types of unit volume each and two vendors with a production cost of $0$. We use the terms left and right to refer to the vendors and buyer types. The valuation of the left buyers is $v_{\ell \ell}=5$ for the left vendor and $v_{\ell r}=3$ for the right vendor; the valuation of the right buyers is $v_{r\ell}=3$ for the left vendor and $v_{r r}=5$ for the right vendor. We show that no pair of a price vector (with prices $p_\ell$ and $p_r$ for the left and the right vendor, respectively) and consistent assignment can be an equilibrium. We distinguish between cases.

If only the left vendor is assigned buyers (the argument holds symmetrically for the right vendor), the right vendor can attract the right buyers by setting its price to any non-zero value smaller than $2$ and increase its utility from $0$ to positive.

The case where the left vendor gets no left buyers and the right vendor gets no right buyers requires the inequalities $v_{\ell \ell}-p_\ell\leq (v_{\ell r}-p_r)^+$ and $v_{r r}-p_r\leq (v_{r\ell}-p_\ell)^+$, i.e., $5- p_\ell \leq (3-p_r)^+$ and $5- p_r \leq (3-p_\ell)^+$; this implies that $p_\ell\geq 5$ and $p_r\geq 5$, which in turn implies that no vendor receives any buyers since they all have negative utility and prefer to abstain. Clearly, this is not an equilibrium since a vendor can attract all buyers, by lowering its price to smaller than $5$, and, thus, obtain strictly positive utility.

The only remaining case is when some left buyers are assigned to the left vendor and some right buyers are assigned to the right vendor. This requires the two inequalities $v_{\ell\ell}-p_\ell\geq (v_{\ell r}-p_r)^+$ and $v_{r r}-p_r\geq (v_{r\ell}-p_\ell)^+$. Clearly, if the first inequality is strict, the left vendor can increase $p_\ell$ (and, consequently, its utility) so that the first inequality becomes equality. Similarly, the second inequality should be an equality as well. So, the requirements for this case are actually the equalities $5-p_\ell = (3-p_r)^+$ and $5-p_r=(3-p_\ell)^+$. The only prices that satisfy both equalities simultaneously are $p_\ell=p_r=5$ which do not form an equilibrium since any vendor can set its price to a value negligibly smaller than $3$, attract all buyers, and increase its utility from $5$ to arbitrarily close to $6$. The proof is complete.
\end{proof}

We remark that Meir et al.~\cite{MLTB14} also present a two-vendor four-buyer-type price competition game that does not admit any equilibrium; the game in the proof of Lemma \ref{lem:existence2} is the simplest one with this property.

\section{Complexity of equilibria}\label{sec:complexity}
We begin the discussion of this section by formulating some concrete computational problems related to equilibria of price competition games.

\begin{quote}
{\sc VerifyEquilibrium}: Given a price vector $\pp$ and a buyers-to-vendors assignment $\xx$ in a price competition game ${\cal G}$, decide whether $(\xx,\pp)$ is an equilibrium of ${\cal G}$.
\end{quote}

\begin{quote}
{\sc ComputePrice}: Given a buyers-to-vendors assignment $\xx$ in a price competition game ${\cal G}$, decide whether there exists a price vector $\pp$ to which $\xx$ is consistent so that $(\xx,\pp)$ is an equilibrium of ${\cal G}$.
\end{quote}

\begin{quote}
{\sc PriceCompetition}: Decide whether a given price competition game has any equilibrium or not.
\end{quote}

\subsection{Complexity of {\sc VerifyEquilibrium}}

{\sc VerifyEquilibrium} can be easily seen to be solvable in $\bigO(nm\log n)$ time. First, one needs to check whether $\xx$ is consistent to $\pp$, i.e., whether the utility of each buyer type is maximized at the vendor(s) used in $\xx$; this can be done by computing at most $\bigO(nm)$ buyer utilities. Then, for every vendor $j$, it suffices to sort (in $\bigO(n\log{n})$ time) the buyer types according to the maximum price that would incentivize them to select vendor $j$, given the prices for the remaining vendors, and, then, compute (in $\bigO(n)$ time) the price that maximizes the vendor's utility when deviating to this price level (a vendor's utility is equal to the volume of buyers it attracts times the difference of the price level from the production cost). The final decision is YES if $\xx$ is consistent to $\pp$ and the utility of all vendors in $(\xx,\pp)$ is equal to the maximum utility over all the deviations considered; otherwise, it is NO. In the following, we call this algorithm \texttt{Verify}; its pseudocode appears below as Algorithm \ref{alg:verify}.

\IncMargin{1em}
\begin{algorithm}
\SetKwInOut{Input}{input}\SetKwInOut{Output}{output}
\Input{A price competition game ${\cal G}$, a price vector $\pp$ and a buyers-to-vendors assignment $\xx$}
\Output{Whether $(\xx,\pp)$ is an equilibrium for ${\cal G}$}
\BlankLine
\ForEach{buyer type $i$}{$D_i(\pp) \leftarrow \arg\max_{j\in M}{\{v_{ij}-p_j\}}$
}
\ForEach{buyer type $i$}{
\ForEach{vendor $j\notin D_i(\pp)$}{\If{$x_{ij}>0$}{ \Return{``$(\xx,\pp)$ is not an equilibrium for ${\cal G}$''}}
}
}
\ForEach{vendor $j$}{
\ForEach{buyer type $i$}{
$q_i \leftarrow  \arg\max_t{\{v_{ij}-t>\max_{j'\neq j}{\{v_{ij'}-p_{j'}\}}\}}$
}
$\pi \leftarrow$ indices of $\qq$ in descending order \\
$vol \leftarrow 0$\\
\For{i $\leftarrow$ 1 \emph{\KwTo} n}{
$vol \leftarrow vol + \mu_{\pi(i)}$\\
\If{$vol\cdot (q_{\pi(i)}-c_j) > u_j(\xx,\pp)$}{\Return {``$(\xx,\pp)$ is not an equilibrium for ${\cal G}$''}}
}
}
\Return{``$(\xx,\pp)$ is an equilibrium for ${\cal G}$''}
\caption{\texttt{Verify}\label{alg:verify}}
\end{algorithm}\DecMargin{2em}

\subsection{Complexity of {\sc ComputePrice}}

The problem {\sc ComputePrice} looks significantly more difficult at first glance since there are too many price vectors (to which $\xx$ is consistent) that have to be considered. Interestingly, we will present a polynomial-time algorithm (henceforth called \texttt{CandidatePrice}) which, given a price competition game and a buyers-to-vendors assignment $\xx$, comes up with a single candidate price vector $\pp$ that can in turn easily be checked whether it forms an equilibrium together with $\xx$ using \texttt{Verify}.

\IncMargin{1em}
\begin{algorithm}
\SetKwInOut{Input}{input}\SetKwInOut{Output}{output}
\Input{A price competition game ${\cal G}$ and a buyers-to-vendors assignment~$\xx$}
\Output{A price vector $\pp$}
\SetKw{Or}{or}
\BlankLine
\tcc*[h]{Create the directed graph $H$}\\
$V(H) \leftarrow \emptyset$, $E(H) \leftarrow \emptyset$\\
\ForEach{vendor $j$}{$V(H) \leftarrow  V(H)\cup \{j\}$}
\ForEach{vendor $j$ and buyer type $i$ with $x_{ij}>0$}{
\ForEach {vendor $j'$}{ \If{$v_{ij}-c_j\leq v_{ij'}-c_{j'}$}{$E(H) \leftarrow  E(H) \cup \{(j,j')\}$} }
}
\tcc*[h]{Compute the set $Z$ of seed vendors}\\
$Z \leftarrow  \emptyset$\\
\Repeat{$Z$ does not change during a round}{
\ForEach{vendor $j$}{
Set $B(j) \leftarrow  \{i: x_{ij}>0\}$\\
\If{$B(j) = \emptyset$ \Or $\min_{i\in B(j)}{v_{ij}} = c_j$ \Or $j$ is in a directed cycle of $H$ \Or $\exists j'\in Z: (j,j')\in H$}{$Z \leftarrow  Z \cup \{j\}$}
}
}
\tcc*[h]{Compute the price vector}\\
\ForEach{seed vendor $j\in Z$}{
$p_j \leftarrow  c_j$
}
\ForEach{non-seed vendor $j\not\in Z$}{
\begin{eqnarray*}
p_j \leftarrow  \left\{\begin{array}{l l} \min\limits_{i:x_{ij}>0}{v_{ij}}, & \mbox{if } Z=\emptyset\\
\min\limits_{i:x_{ij}>0}{\left\{v_{ij}-\max\limits_{j'\in Z}{\{v_{ij'}-c_{j'}\}^+}\right\}}, & \mbox{otherwise} \end{array}\right.
\end{eqnarray*}
}
\caption{\texttt{CandidatePrice}\label{alg:candidateprice}}
\end{algorithm}\DecMargin{1em}

\texttt{CandidatePrice} works as follows; see also Algorithm \ref{alg:candidateprice}. It first computes a set $Z$ of {\em seed} vendors which will have a price equal to their production cost. In order to define $Z$, it is convenient to consider the directed graph $H$ that has a node for each vendor and a directed edge from node $j$ to node $j'$ labelled by $i$ if buyers of type $i$ are assigned to vendor $j$ in $\xx$ and, furthermore, $v_{ij}-c_j \leq v_{ij'}-c_{j'}$. Now, the set $Z$ is defined recursively as follows:
\begin{enumerate}
\item Any vendor that is not assigned any buyer in $\xx$ belongs to $Z$; such a vendor is called {\em empty}.
\item Any vendor $j$ such that $\min_{i:x_{ij}>0}{v_{ij}}=c_j$ belongs to $Z$.
\item Any vendor that is part of a directed cycle in $H$ belongs to $Z$.
\item Any vendor that has a directed edge to a vendor of $Z$ also belongs to $Z$.
\end{enumerate}
\texttt{CandidatePrice} returns the price vector $\pp$ with $p_j=c_j$ for each seed vendor $j$ and
\begin{eqnarray}\label{eq:candidate-price}
p_j = \left\{\begin{array}{l l} \min\limits_{i:x_{ij}>0}{v_{ij}}, & \mbox{if } Z=\emptyset\\
\min\limits_{i:x_{ij}>0}{\left\{v_{ij}-\max\limits_{j'\in Z}{\{v_{ij'}-c_{j'}\}^+}\right\}}, & \mbox{otherwise} \end{array}\right.
\end{eqnarray}
for each non-seed vendor.

\begin{example}\label{ex:computeprice}
Consider the price competition game ${\cal G}$ shown in Table \ref{tab:example2}.

\begin{table}[htbp]
\centerline{
\begin{tabular}{l|c c c c}
  & vn $1$& vn $2$& vn $3$& vn $4$\\[-0.5em]
  & $c_1=0$ & $c_2=0$& $c_3=0$& $c_4=0$\\\hline
  bt $A$: $\mu_A=1$ & \circled{$1$} & $1$ & $0$ &$1$\\
  bt $B$: $\mu_B=1$ & $1$ & $0$ & $0$ & \circled{$2$} \\
  bt $C$: $\mu_C=1$ & $0$ & \circled{$1$} & $1$ & $0$\\
  bt $D$: $\mu_D=1$ & \circled{$1$} & $0$ & $1$ & $0$\\
  bt $E$: $\mu_E=1$ & $1$ & $0$ & \circled{$1$} & $0$
\end{tabular}}
  \caption{The price competition game discussed in Example \ref{ex:computeprice}. The assignment $\xx$ is denoted by circles.} \label{tab:example2}
  \end{table}

Let $\xx$ be the assignment where all buyers of type $A$ are assigned to vendor $1$, all buyers of type $B$ are assigned to vendor $4$, buyers of type $C$ are assigned to $2$, buyers of type $D$ are assigned to vendor $1$ and, finally, buyers of type $E$ are assigned to vendor $3$. The directed graph $H$ is depicted in Figure \ref{fig:computeprice}.

Note that vendors $1$, $2$, and $3$ are seed vendors since they are part of a directed cycle. This leads to the following price vector $\pp = \{0,0,0,1\}$. In order to verify that $(\xx,\pp)$ is indeed a pure Nash equilibrium, it suffices to run algorithm \texttt{Verify} with input $(\xx,\pp)$; this is true since all buyers are assigned to vendors in their demand set, vendors $1$, $2$ and $3$ cannot increase their utility by unilaterally increasing their price as, then, they would lose all assigned buyers, while vendor $4$ cannot increase its price, since its assigned buyers would move to vendor $1$, and cannot attract new buyers by decreasing its price.

\begin{figure}[htbp]
\centering
\includegraphics[trim = 0 0 0 0, width = 1.4in]{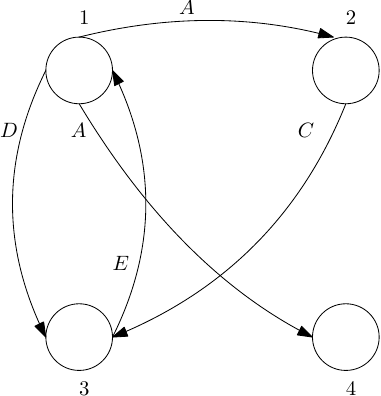}
\caption{The directed graph $H$ arising from the assignment $\xx$ and the price competition game of Table \ref{tab:example2}.} \label{fig:computeprice}
\end{figure}
\end{example}

The correctness of the algorithm is given by the following lemma.
\begin{lemma}\label{lem:candidate}
Let $\pp$ be the price vector returned by \texttt{CandidatePrice} on input a price competition game ${\cal G}$ and a buyers-to-vendors assignment $\xx$. If ${\cal G}$ admits an equilibrium $(\xx,\qq)$, then $q_j=p_j$ for every non-empty vendor $j$.
\end{lemma}

\begin{proof}
Consider the price vector $\pp$ returned by \texttt{CandidatePrice}. We will show that $\qq$ satisfies the very same properties as $\pp$ does, for non-empty seed and non-seed vendors. We begin by arguing about the empty vendors, where it may hold that $p_j\neq q_j$ for an empty vendor $j$.

\paragraph{Empty vendors}
Given the price vector $\qq$, we construct a new price vector $\bar{\qq}$ by setting $\bar{q}_j=c_j$ for every empty vendor $j$ in $\xx$ and $\bar{q}_j=q_j$ for every other vendor $j$. Since $(\xx,\qq)$ is an equilibrium, $(\xx,\bar{\qq})$ is an equilibrium as well. To see why, it suffices to show that no buyer has any incentive to deviate to an empty vendor in $\xx$. Assume otherwise that $(\xx,\bar{\qq})$ is not an equilibrium; we will show that $(\xx,\qq)$ should not be an equilibrium as well. Since $(\xx,\bar{\qq})$ is not an equilibrium, there must exist some $i$ such that buyers of type $i$ are assigned to a vendor $j$ in $\xx$ and satisfy $v_{ij'}-\bar{q}_{j'} > v_{ij}-\bar{q}_{j}$ for some vendor $j'$ that is empty in $\xx$. From the definition of $\bar{\qq}$, this means that $v_{ij'}-c_{j'} > v_{ij}-q_{j}$ and vendor $j'$ could set its price to a sufficiently small but strictly higher than its production cost price to attract the buyers of type $i$ and increase its utility from $0$ (since $j'$ is an empty vendor according to $\xx$) to positive; this contradicts the assumption that $(\xx, \qq)$ is an equilibrium. Now, clearly, $\bar{\qq}$ has the desired structure (i.e., prices equal to production costs) for vendors that are empty in $\xx$, and it suffices to prove that $\pp=\bar{\qq}$.

For the remainder of the proof, recall that, by the definition of $\qq$, $q_j = \bar{q}_j$ for any non-empty vendor $j$.

\paragraph{Seed vendors}
First, observe that for any non-empty vendor $j$ with $\min_{i:x_{ij}>0}{v_{ij}}=c_j$, it should obviously be $\bar{q}_j=c_j$ since $\xx$ is consistent to $\bar{\qq}$.

Now, consider a directed cycle $L$ of length $\ell$ in $H$; we have to show that $\bar{q}_j=c_j$ for each vendor $j\in L$. Let $j_1$, $j_2$, ..., $j_{\ell}$ be the vendors in $L$ and denote by $i_1$, $i_2$, ..., $i_{\ell}$ the (not necessarily distinct) labels of the corresponding edges. Observe that the (consistency) conditions that guarantee that the buyers of type $i_k$ do not prefer vendor $j_{k+1}$ to vendor $j_k$ for $k \in [\ell-1]$ and buyers of type $i_{\ell}$ do not prefer vendor $j_1$ to vendor $j_{\ell}$ together with the conditions defining the directed cycle $L$ can only hold if $v_{i_k j_k}-\bar{q}_{j_k} = v_{i_{k}j_{k+1}}-\bar{q}_{j_{k+1}}$ for all $k\in [\ell-1]$ and $v_{i_{\ell} j_{\ell}}-\bar{q}_{j_{\ell}} = v_{i_{\ell}j_1}-\bar{q}_{j_1}$. Now, if there is a vendor $j'\in L$ with $\bar{q}_{j'}>c_{j'}$, let $j^*$ be the predecessor of $j'$ in $L$. Then, by negligibly decreasing $\bar{q}_{j'}$, vendor $j'$ could also attract the buyers assigned to both vendor $j^*$ and itself in $\xx$; this would strictly increase its utility, contradicting the fact that $(\xx,\bar{\qq})$ is an equilibrium.

We will also show that for every vendor $j$ that is identified as a seed vendor by the recursive step 4 of \texttt{CandidatePrice}, it is $\bar{q}_j=c_j$. Assume otherwise and consider the first vendor $j$ that is identified as seed vendor by the recursive step 4 that has $\bar{q}_j>c_j$. This means that $H$ has an edge from $j$ to some seed vendor $j'$, labelled by $i$, and, hence, $v_{ij}-c_j \leq v_{ij'}-c_{j'}$. By our assumption that $j$ is the first seed vendor with $\bar{q}_j>c_j$ identified at step 4 of the algorithm, it is $\bar{q}_{j'}=c_{j'}$. Hence, $v_{ij}-\bar{q}_j < v_{ij'}-\bar{q}_{j'}$ which contradicts the consistency of $\xx$ to~$\bar{\qq}$.

\paragraph{Non-seed vendors}
First, we will prove that for every such vendor $j$, it must be $\bar{q}_j>c_j$. Consider the subgraph $H'$ that is induced by the nodes of $H$ corresponding to non-seed vendors. By the definition of seed vendors, $H'$ is acyclic and its vendors are non-empty. We will show that the inequality holds for the leaves of $H'$ and will then use this fact to also show that it holds for their parents and, recursively, for all vendors in $H'$.

\paragraph{Case I: leaf vendors} For a leaf vendor $j$ of $H'$, its price may satisfy $\bar{q}_j=c_j$ only if there is a buyer type $i$ assigned to $j$ in $\xx$ and another vendor $j'$ so that $v_{ij}-\bar{q}_j = v_{ij'}-\bar{q}_{j'}$, as otherwise $j$ could negligibly increase its price and obtain more utility.
But then, it is $v_{ij}-c_j=v_{ij'}-\bar{q}_{j'}\leq v_{ij'}-c_{j'}$ and, hence, $H$ should also contain the directed edge from $j$ to $j'$. Since $j$ is non-seed, $j'$ should be non-seed as well and the edge should exist in $H'$, contradicting the fact that $j$ is a leaf in $H'$.

\paragraph{Case II: non-leaf vendors} Let $j^*$ be a non-seed vendor with $\bar{q}_{j^*}=c_{j^*}$ such that all its children $j$ in $H'$ have $\bar{q}_j>c_j$. Again, there is a buyer type $i$ assigned to $j^*$ in $\xx$ and another vendor $j'$ so that $v_{ij^*}-\bar{q}_{j^*} = v_{ij'}-\bar{q}_{j'}$, as, otherwise, vendor $j^*$ could increase its price above its production cost and increase its utility. Since $\bar{q}_{j^*} = c_{j^*}$ and $\bar{q}_{j'}\geq c_{j'}$, our assumptions imply that $v_{ij^*}-c_{j^*} \leq v_{ij'}-c_{j'}$. Hence, $j'$ is one of the children of $j^*$. But then, $j'$ can negligibly decrease its price to attract (in addition to all the buyers it gets in $\xx$) the buyers assigned to $j^*$ as well; this would increase the utility of vendor $j'$. This contradicts the fact that $(\xx,\bar{\qq})$ is an equilibrium and establishes that $\bar{q}_j>c_j$ for every non-seed vendor $j$.

It remains to show that for every non-seed vendor $j$, $\bar{q}_j$ is given by the same expression (\ref{eq:candidate-price}) as $p_j$. If there are no seed vendors, i.e., (the set $Z$ is empty), assume that there is a vendor $j$ with $\bar{q}_j<\min_{i:x_{ij}>0}{v_{ij}}$. Then, the only reason that may prevent vendor $j$ from increasing its price is that there is a buyer of type $i$ assigned to $j$ in $\xx$ and another vendor $j'$ such that $v_{ij}-\bar{q}_j=v_{ij'}-\bar{q}_{j'}$. Since $j'$ is a non-seed vendor, it is $\bar{q}_{j'}>c_{j'}$ and, hence, vendor $j'$ could negligibly decrease its price to attract (in addition to all the buyers it gets in $\xx$) the buyers assigned to $j$ as well; this would increase the utility of vendor $j'$.

Otherwise, if the set $Z$ of seed vendors is non-empty, we will show that every non-seed vendor $j$ has $\bar{q}_j=
p_j$. First, assume otherwise that $\bar{q}_j>p_j$. By the definition of $p_j$, this means that there is a buyer type $i$ assigned to $j$ in $\xx$ and a seed vendor $j'$ such that $\bar{q}_j > v_{ij}-(v_{ij'}-c_{j'})^+$ and subsequently (since $\bar{q}_{j'}=c_{j'}$), $v_{ij}-\bar{q}_j < (v_{ij'}-q_{j'})^+$. This contradicts the consistency of $\xx$ to $\bar{\qq}$. Finally, assume that $\bar{q}_j<p_j$. Then, the only reason that may prevent vendor $j$ from increasing its price is that there is a buyer of type $i$ assigned to $j$ in $\xx$ and another vendor $j'$ such that $v_{ij}-\bar{q}_j=v_{ij'}-\bar{q}_{j'}$. If $j'\in Z$, the last equality also implies that $\bar{q}_j = p_j$, contradicting our assumption. Hence, it must be $j'\not\in Z$, i.e., $\bar{q}_{j'}>c_{j'}$. Then, vendor $j'$ could negligibly decrease its price to attract (in addition to the buyers it gets in $\xx$) the buyers assigned to $j$ as well; this would increase the utility of vendor $j'$. This completes the proof that $\pp = \bar{\qq}$, and, hence, that $p_j = q_j$ for any non-empty vendor $j$.
\end{proof}

\subsection{Complexity of {\sc PriceCompetition}}

We now turn our attention to {\sc PriceCompetition} and first consider the cases where either the number of buyer types or the number of vendors is constant.

\begin{theorem}{\sc PriceCompetition} is solvable in polynomial time when the number of buyer types is constant.
\end{theorem}
\begin{proof}
We begin by arguing that it suffices to consider non-fractional assignments, i.e., assignments where all buyers of the same type are assigned to a single vendor. To see that, consider an equilibrium $(\xx,\pp)$ where buyers of some type $i$ are split between two vendors $j$ and $j'$. Then, it must be $p_j=c_j$ and $p_{j'}=c_{j'}$, and the assignment in which all these buyers are assigned to vendor $j$ without changing the prices is still an equilibrium. Indeed, at equilibrium the utilities of buyers of type $i$ assigned to $j$ and $j'$ should be the same. Hence, if one of the two vendors had a price strictly higher than its production cost, it could increase its utility by negligibly decreasing its price; this would result in attracting all buyers of type $i$ previously assigned to $j$ and $j'$. So, by moving all buyers of type $i$ from vendor $j'$ to vendor $j$, we still have an assignment that is consistent to $\pp$ in which the utilities of buyers and vendors do not change; observe that since $p_{j'} = c_{j'}$, it holds that $u_{j'}(\xx,\pp)=0$ and, hence, vendor $j'$ does not lose utility by this change. Clearly, the new assignment is also an equilibrium.

The theorem follows since there are at most $(m+1)^n$ different non-fractional buyers-to-vendors assignments and corresponding instances of {\sc ComputePrice} that we need to consider.
\end{proof}

The case of a constant number of vendors (this can be thought of as an {\em oligopoly}) is considerably more involved but still computable in polynomial time as we show in the following.

For a fixed price vector $\pp$, we define the induced preference, denoted by $\succ_i$, of buyers of type $i$ as $j\succ_i j'$ if $v_{ij}-p_j>v_{ij'}-p_{j'}$ and $j\simeq_i j'$ if $v_{ij}-p_j = v_{ij'}-p_{j'}$. We use the term preference profile to refer to a combination of buyer preferences. The main idea is to enumerate the different preference profiles that are defined for all price vectors $\pp\in \R^m$. Observe that the sign (from $\{-,0,+\}$) of the expression $v_{ij}-v_{ij'}-p_j+p_{j'}$ indicates whether buyers of type $i$ prefer vendor $j$ to vendor $j'$ (i.e., $j\succ_i j'$), are indifferent between the two (i.e., $j\simeq j'$), or prefer vendor $j'$ to $j$ (i.e., $j' \succ_i j$). The number of different preferences of buyers between the two specific vendors, when keeping valuations fixed, is given by the number of different sign patterns for the expressions $v_{ij}-v_{ij'}-p_j+p_{j'}$ for $i=1, ..., n$ as the difference $p_j-p_{j'}$ runs from $-\infty$ to $\infty$. Since there are at most $n$ different values of the difference $v_{ij}-v_{ij'}$, this number is at most $2n+1$; this holds since, at the beginning, all signs are positive and, as $p_j-p_{j'}$ increases, the sign of each buyer type will change, first to $0$ and then to negative. In total, the number of distinct sign patterns we need to enumerate in order to consider all distinct preference profiles is at most $(2n+1)^{{m+1\choose 2}}$ as there are ${m+1 \choose 2}$ ways to select a pair $j$, $j'$ of vendors; this is polynomial in $n$ when $m$ is constant.

When considering a preference profile $\succeq$, we compute the following assignment $\xx$ which should be given to \texttt{CandidatePrice} in order to return a price vector $\pp$; the pair $(\xx,\pp)$ will in turn be given to \texttt{Verify} to detect whether it corresponds to an equilibrium or not. For each buyer type $i$ with a unique top preference (i.e., strictly preferring a particular vendor to all others), $\xx$ assigns buyers of type $i$ to their most preferred vendor. For each buyer type $i$ that has a set $T$ of at least two vendors tied as its top preference, $\xx$ assigns $i$ to a(ny) vendor $j$ of $T$ maximizing $v_{ij}-c_j$. We call this algorithm \texttt{Enumerate}; its pseudocode appears as Algorithm \ref{alg:enumerate}.

\IncMargin{1em}
\begin{algorithm}
\SetKwInOut{Input}{input}\SetKwInOut{Output}{output}
\SetKw{CandidatePrice}{CandidatePrice}
\SetKw{Verify}{Verify}
\Input{A price competition game ${\cal G}$}
\Output{Whether ${\cal G}$ admits an equilibrium}
\BlankLine
\tcc*[h]{A preference profile is a vector of preferences, i.e., of rankings over vendors, for each buyer type}\\
\ForEach{valid preference profile $\succeq$}{
\ForEach{buyer type $i$}{
$T(i) \leftarrow$ top preference in $\succeq_i$\\
\If{$|T(i)|=1$}{$x_{iT(i)} = \mu_i$}
\Else{
$j' \leftarrow \arg\max_{j \in T(i)}\{v_{ij}-c_j\}$\\
$x_{ij'} = \mu_i$
}
}
$\pp \leftarrow \CandidatePrice({\cal G}, \xx)$\\
\Verify$({\cal G}, \xx,\pp)$\\
}
\caption{\texttt{Enumerate}\label{alg:enumerate}}
\end{algorithm}\DecMargin{1em}

Clearly, on input a price competition game that does not admit an equilibrium, \texttt{Enumerate} will not find any. The next lemma completes the proof of correctness of \texttt{Enumerate}.
\begin{lemma}\label{lem:enumerate}
On input a price competition game, \texttt{Enumerate} returns an equilibrium if one exists.
\end{lemma}

\begin{proof}
Assume that \texttt{Enumerate} is applied on input a price competition game ${\cal G}$ that admits an equilibrium $(\xx,\qq)$. If $\xx$ is the unique assignment that is consistent to $\qq$, \texttt{Enumerate} will consider the preference profile $\succeq$ corresponding to vector $\qq$ and will pass the uniquely defined assignment $\xx$ as input to \texttt{CandidatePrice} to compute a price vector $\pp$; by Lemma \ref{lem:candidate}, $(\xx,\pp)$ will form an equilibrium of ${\cal G}$.

Now, assume that $\xx$ is not the unique assignment that is consistent to $\qq$. Denote by $\bar{\xx}$ the  assignment computed by \texttt{Enumerate} (notice that $\bar{\xx}$ is consistent to $\qq$ as well) when considering the preference profile that corresponds to the price vector $\qq$. We will show that $(\bar{\xx},\qq)$ is an equilibrium as well; then, Lemma \ref{lem:candidate} guarantees that an equilibrium will be found when \texttt{Enumerate} will run \texttt{CandidatePrice} with input assignment $\bar{\xx}$.

Consider a buyer type $i$ with $x_{ij}>0$ and $\bar{x}_{ij'}>0$ for two different vendors $j$ and $j'$. By our assumptions, we have $v_{ij}-q_j=v_{ij'}-q_{j'}$ (since buyers of type $i$ are indifferent between vendors $j$ and $j'$ in $\qq$) and $v_{ij}-c_j \leq v_{ij'}-c_{j'}$ (since \texttt{Enumerate} sets $\bar{x}_{ij'}>0$). We will show that $q_j=c_j$ and $q_{j'}=c_{j'}$. Indeed, assume that $q_{j'}> c_{j'}$. Then, by negligibly decreasing its price in $\qq$, vendor $j'$ could increase its utility by attracting (in addition to the buyers it gets in $\xx$) all the buyers of type $i$. Hence, $q_{j'}=c_{j'}$ and $v_{ij}-q_j=v_{ij'}-c_{j'}$. By the inequality $v_{ij}-c_j\leq v_{ij'}-c_{j'}$, we obtain that $q_j = c_j$. The lemma follows since the different assignments of buyers in $\xx$ and $\bar{\xx}$ does not affect the utility of the corresponding vendors (which is zero).
\end{proof}

By the discussion above and Lemma \ref{lem:enumerate}, we obtain the following.
\begin{theorem}{\sc PriceCompetition} is solvable in polynomial time when the number of vendors is constant.
\end{theorem}

The restrictions on the numbers of vendors or buyer types are necessary in order to come up with efficient algorithms for {\sc PriceCompetition} (unless $\mbox{P}=\mbox{NP}$).

\begin{theorem}\label{thm:pc-hardness}
{\sc PriceCompetition} is NP-hard.
\end{theorem}

\begin{proof}
We use a reduction from the well-known NP-hard {\sc Exact-$3$-Cover} problem (X$3$C) which is formally described as follows.
\begin{quote}
{\sc Exact-3-Cover}: Given a universe $\UU$ of $n=3q$ elements and a collection $\CC$ of $m$ $3$-sized subsets of $\UU$, is there a cover $C\subseteq \CC$ consisting of $q$ disjoint sets?
\end{quote}
Given an instance of X$3$C we describe an instance of {\sc PriceCompetition} as follows; see also Table~\ref{tab:pc-hardness-instance}:
\begin{itemize}
\item For every subset $S \in \CC$, there is a set-vendor $\nu_S$;
\item for every element $e \in \UU$, there are two vendors $\nu_{e,1}$ and $\nu_{e,2}$;
\item for every element $e \in \UU$, there is an element-buyer type $b_e$ with volume $1$, and valuations $12$ for every vendor $\nu_S$ such that $e \in S$, $5$ for vendor $\nu_{e,1}$, and $3$ for vendor $\nu_{e,2}$;
\item for every element $e \in \UU$, there is a buyer type $b^*_{e}$ with volume $1$ and valuation $3$ for vendor $\nu_{e,1}$, and $5$ for vendor $\nu_{e,2}$;
\item for every subset $S \in \CC$, there are two buyer types: one buyer type $b_{S,h}$ with volume $3$ and valuation $30$ for vendor $\nu_S$, and one buyer type  $b_{S,\ell}$ with volume $9$ and valuation $6$ for vendor $\nu_S$;
\item all valuations not mentioned above and all production costs are zero.
\end{itemize}

\begin{table}[htbp]
\centerline{
\begin{tabular}{l|c c c}
  & vn $\nu_S$& vn $\nu_{e,1}$& vn $\nu_{e,2}$\\[-0.5em]
  & $c_{\nu_S}=0$ & $c_{\nu_{e,1}}=0$ & $c_{\nu_{e,2}}=0$ \\\hline
  bt $b_{e}$: $\mu_{b_{e}}=1$ & $12$ & $5$ & $3$ \\
  bt $b^*_{e}$: $\mu_{b^*_{e}} = 1$ & $0$  & $3$ & $5$ \\
  bt $b_{S,h}$: $\mu_{b_{S,h}} = 3$ & $30$ & $0$ & $0$ \\
  bt $b_{S,\ell}$: $\mu_{b_{S,\ell}} = 9$ & $6$  & $0$ & $0$ \\
\end{tabular}}
  \caption{Part of the instance in the reduction from {\sc Exact-$3$-Cover}.}
\label{tab:pc-hardness-instance}
\end{table}

First, observe that in any equilibrium $(\xx,\pp)$, a vendor $\nu_S$ corresponding to the set $S=\{e_1,e_2,e_3\}$ of ${\cal C}$ may have been assigned
\begin{itemize}
\item either all buyers of type $b_{S,h}$ by setting its price $p_{\nu_S}$ to $30$, i.e., equal to the valuation $v_{b_{S,h},\nu_S}$;
\item or all buyers of types $b_{S,h}$, $b_{e_1}$, $b_{e_2}$, $b_{e_3}$, and $b_{S,\ell}$ by setting its price to $6$, i.e., equal to the valuation $v_{b_{S,\ell},\nu_S}$.
\end{itemize}
To see why the above claim is true, observe that, in any of the two cases, the utility of vendor $\nu_S$ at equilibrium is $90$. Clearly, by decreasing its price so that it lies in the ranges $(12, 30)$ or $(0,6)$, the utility of the vendor becomes strictly smaller than $90$ since any consistent assignment will belong to one of the two cases. By setting the price in the range $(6,12]$, no consistent assignment includes the buyers of type $b_{S,\ell}$ and, hence, the utility of vendor $\nu_S$ is at most $72$.

Now, assume that the initial X3C instance is a YES instance with a cover $C$ consisting of $q$ $3$-sets from ${\cal C}$. Then, consider the price vector $\pp$ with $p_{\nu_S} = v_{b_{S,\ell},\nu_S}=6$ for every set $S\in C$, $p_{\nu_S} = v_{b_{S,h},\nu_S}=30$ for every set $S\not\in C$, and $p_{\nu_{e,1}}=0$ and $p_{\nu_{e,2}}=2$ for every $e\in \UU$. Let $\xx$ be the consistent assignment that assigns each element buyer of type $e$ to the set vendor $\nu_S$ such that $S\in C$ and $e\in S$, assigns the buyers of type $b_{S,h}$ to vendor $\nu_S$ for every set $S\in {\cal C}$, assigns the buyers of type $b_{S,\ell}$ to
vendor $\nu_S$ for every set $S\in C$ while the buyers of type $b_{S,\ell}$ abstain for every set $S\not\in C$, and assigns the buyers of type $b^*_{e}$ to vendor $\nu_{e,2}$. We now argue that $(\xx,\pp)$ is an equilibrium. First, observe that all set vendors follow one of the two strategies mentioned above and are at equilibrium. Now, for every element $e\in \UU$, the vendors $\nu_{e,1}$ and $\nu_{e,2}$ are essentially engaged in a two-vendor game competing only for buyers of type $b^*_e$ since the buyers of type $b_e$ have a utility of $6$ at vendor $\nu_S$ with $e\in S$ and their deviation to a vendor among $\nu_{e,1}$ and $\nu_{e,2}$ would never increase their utility to above $5$. The prices $p_{\nu_{e,1}}=0$ and $p_{\nu_{e,2}}=2$ form an equilibrium in this two-vendor game.  Indeed, the utility of vendor $\nu_{e,1}$ is zero (since it is not assigned any buyers according to $\xx$) and its price equals its production cost, which means that this vendor cannot gain any utility by deviating. Moreover, the utility of vendor $\nu_{e,2}$ is $2$ and it cannot be increased (since the buyers of type $b_{e^*}$ would then go to vendor $\nu_{e,1}$ or abstain) or decreased (since the vendor would then have less utility).

On the other hand, if the initial X3C instance is a NO instance, we claim that no assignment $\xx$ of buyers to vendors can be consistent to a price vector $\pp$ so that $(\xx,\pp)$ is an equilibrium. Assume otherwise and consider all set vendors. By the discussion above, each set vendor has set its price either to $30$ or to $6$ at equilibrium. We remark that there exists at least one element-buyer type $b_e$ that has no buyers assigned to any set vendor $\nu_S$ that corresponds to a subset $S$ containing element $e \in \UU$; this holds since the initial X3C instance is a NO instance. Therefore, all these set vendors, that are candidate vendors for buyers of type $b_e$, have set their price to exactly $30$ in the equilibrium and, so, the buyers of type $b_e$ can only be assigned to vendors $\nu_{e,1}$ or $\nu_{e,2}$. Hence, the vendors $\nu_{e,1}$ and $\nu_{e,2}$ can compete not only for the buyers of type $b_e^*$ but also for the buyers of type $b_e$. But then, observe that these two vendors and the buyers of type $b_e$ and $b_e^*$ would follow the circular dynamics in the counter-example of Lemma \ref{lem:existence2}; this implies that no equilibrium exists and the proof is complete.
\end{proof}

\section{Enforcing equilibria using subsidies}\label{sec:subsidies}

We now consider the option to use subsidies. A subsidy given to a vendor aims to compensate it for setting its price at a particular value. In this way, subsidies can be used to enforce a particular pair of price vector and consistent buyers-to-vendors assignment. Formally, given a price vector $\pp$ and a consistent assignment $\xx$, denote by $\theta_j(\xx,\pp)$ the maximum utility of vendor $j$ over all deviations $p'_j$ and all assignments $\yy$ that are consistent to $(p'_j,\pp_{-j})$. Vendor $j$ has no incentive to follow any such deviation when it is given an amount of subsidies $s_j\geq \theta_j(\xx,\pp)-u_j(\xx,\pp)$. If this inequality holds for every vendor $j$, we say that the pair $(\xx,\pp)$ is enforced as an equilibrium. We denote by $\s(\xx,\pp)$ the entry-wise minimum subsidy vector that enforces $(\xx,\pp)$ as an equilibrium, i.e., $s_j(\xx,\pp) = \theta_j(\xx,\pp)-u_j(\xx,\pp)$. We use the terms ``total amount'' and ``cost'' to refer to the sum of all entries of a subsidy vector.

Our first observation is that a large amount of subsidies may inherently be necessary in order to enforce {\em any} equilibrium.

\begin{theorem}
For every $\delta>0$, there exists a price competition game, in which no subsidy assignment of cost smaller than $(1/4-\delta)\SW^*$ can enforce any pair of price vector and consistent buyers-to-vendors assignment as an equilibrium.
\end{theorem}

\begin{proof}
Consider a price competition game with two buyer types of unit volume each and two vendors with a production cost of $0$. We use the terms left and right to refer to the vendors and buyer types. Let $\epsilon>0$; then the valuations are $v_{\ell\ell}=v_{rr}=4-\epsilon$ and $v_{r\ell}=v_{\ell r}=3$. We will show that no price vector and consistent buyers-to-vendors assignment can be enforced as an equilibrium using an amount of subsidies less than $2-2\epsilon$. The optimal social welfare is $8-2\epsilon$ and the theorem will follow by setting $\epsilon$ to a sufficiently small positive value.

Consider a buyers-to-vendors assignment $\xx$ and let $p_\ell$ and $p_r$ be the two prices to which $\xx$ is consistent. First observe that the left vendor will attract buyers of both types by setting its price to (any value that is negligibly smaller than) $\min\{3,p_r-1+\epsilon\}$ (if positive) while it can attract the buyers of left type by setting its price to (any value that is negligibly smaller than) $\min\{4-\epsilon,1-\epsilon+p_r\}$; the two terms in the $\min$ expression denote the price necessary so that the buyers do not abstain and prefer the left vendor to the right one. Hence, by expressing the subsidies in terms of the maximum utility of the vendors at any deviation minus their current utility we get an amount $\beta$ of subsidies that is
\begin{align}\nonumber
\beta &\geq  \max\{\min\{6,2p_r-2+2\epsilon\},\min\{4-\epsilon,1-\epsilon+p_r\}\}-(x_{\ell\ell}+x_{r\ell})p_\ell\\\label{eq:subsidies}
&\quad +\max\{\min\{6,2p_\ell-2+2\epsilon\},\min\{4-\epsilon,1-\epsilon+p_\ell\}\}-(x_{\ell r}+x_{rr})p_r.
\end{align}
We will distinguish between cases (omitting symmetric ones) depending on the price values $p_\ell$ and $p_r$ and will use (\ref{eq:subsidies}) to show that $\beta \geq 2-2\epsilon$ in any case.
\smallskip

\noindent{\em Case 1: $p_\ell,p_r>4-\epsilon$.} All buyers abstain and both $\max$ expressions evaluate to at least $6$.
\smallskip

\noindent{\em Case 2: $p_\ell>4-\epsilon$ and $3<p_r\leq 4-\epsilon$.} The left buyers abstain and the right buyers may either abstain or be assigned to the right vendor. By just considering the last two terms of (\ref{eq:subsidies}), we get $\beta\geq 6-p_r\geq 2+\epsilon$.
\smallskip

\noindent{\em Case 3: $p_\ell>4-\epsilon$ and $p_r\leq 3$.} The right vendor may attract all buyers and the amount of subsidies in it is at least $6-2p_r$. The amount of subsidies at the left vendor is then at least $2p_r-2+2\epsilon$. In total, $\beta\geq 4+2\epsilon$.
\smallskip

\noindent{\em Case 4: $3<p_\ell,p_r\leq 4-\epsilon$.} The left (respectively, right) buyers can only be assigned to the left (respectively, right) vendor. Since the $\max$ expressions of (\ref{eq:subsidies}) evaluate to $2p_r-2+2\epsilon$ and $2p_\ell-2+2\epsilon$, we obtain $\beta \geq  p_r+p_\ell-4+4\epsilon\geq$ $2+4\epsilon$.
\smallskip

\noindent{\em Case 5: $3<p_\ell\leq 4-\epsilon$ and $p_r\leq 3$.} The right hand side of (\ref{eq:subsidies}) is minimized when buyers of different types are assigned to different vendors. The $\max$ expressions are at least $1-\epsilon+p_r$ and at least $2p_\ell-2+2\epsilon$. Altogether, we get $\beta\geq p_{\ell}-1+\epsilon\geq 2+\epsilon$.
\smallskip

\noindent{\em Case 6: $p_\ell,p_r\leq 3$.} We have to distinguish between a few subcases. First, consider the subcase where each vendor receives the buyers of a different type. Clearly, the $\max$ expressions evaluate to at least $1-\epsilon+p_r$ and $1-\epsilon+p_\ell$; altogether, we get $\beta\geq 2-2\epsilon$. If, in contrast, buyers of both types are assigned to the same (say, left) vendor, this should be due to the fact that $p_r\geq 1-\epsilon+p_\ell$ (so that the right buyer prefers the left vendor). Then, the first $\max$ expression is at least $1-\epsilon+p_r\geq 2-2\epsilon+p_\ell$ while the second one is at least $1-\epsilon+p_\ell$. Overall, we have that $\beta\geq 3-3\epsilon$.
\end{proof}

In the following, we restrict ourselves to optimal assignments and show tight bounds on the cost of subsidies that are necessary and sufficient to enforce these assignments as equilibria.
\begin{theorem}
In every price competition game, the optimal assignment can be enforced as an equilibrium using an amount of subsidies that is at most $\SW^*$. This bound is tight. In particular, for every $\epsilon>0$, there exists a price competition game in which the optimal assignment cannot be enforced as an equilibrium with a total amount of subsidies less than $(1-\epsilon)\SW^*$.
\end{theorem}

\begin{proof}
We first prove the upper bound. Consider an optimal assignment $\xx$ with $x_{ij}\in\{0,\mu_i\}$ for every buyer type/vendor pair and the price vector $\pp$ with $p_j=c_j$; clearly, $\xx$ is consistent to $\pp$. By deviating to any other price, vendor $j$ cannot get any buyers that are not assigned to it in the optimal assignment; it can, however, increase its utility by raising the price for those buyers that are already assigned to it. Hence, it suffices to assign a subsidy of $s_j(\xx,\pp)=\sum_{i}{x_{ij}(v_{ij}-c_j)}$ to each vendor $j$; this obviously yields a total amount of subsidies equal to $\SW^*$ that enforces the optimal assignment as an equilibrium, since each vendor obtains the maximum utility from its assigned buyers.

For the lower bound, let $\chi>2$ and consider the price competition game with two buyer types of unit volume and valuations $\chi$ and $1$ for a single vendor of production cost of $0$. The optimal social welfare is $\SW^*=\chi+1$. Observe that the utility of the vendor is maximized to $\chi$ by setting its price to $\chi$ while any price that is consistent to assigning both buyer types to the vendor is at most $1$ for a vendor utility of at most $2$. Hence, the amount of subsidies required to enforce the optimal assignment as an equilibrium is at least $\chi-2$ which becomes at least $(1-\epsilon)\SW^*$ by setting $\chi$ sufficiently large.
\end{proof}

Even though the minimum amount of subsidies that is sufficient to enforce the optimal assignment as an equilibrium can be large in terms of the optimal social welfare, one might hope that it could be efficiently computable. Unfortunately, this is far from true as we show below in Theorem~\ref{thm:subsidies-apx}. Before presenting the theorem, let us formally define the corresponding optimization problem:

\begin{quote}
{\sc MinSubsidies}: Given a price competition game ${\cal G}$ with an optimal assignment $\xx$, compute a price vector $\pp$ that minimizes the cost $\s(\xx,\pp)$ over all price vectors to which $\xx$ is consistent.
\end{quote}

Notice that {\sc MinSubsidies} should return an equilibrium $(\xx,\pp)$ when one exists. This can be efficiently decided using algorithm \texttt{CandidatePrice}.
The hardness of the problem manifests itself in instances that do not admit equilibria. It is not hard to show that, in this case, \texttt{CandidatePrice} does not minimize the amount of subsidies required. This is shown in the following example.

\begin{example}\label{ex:lemma4} Consider the price competition game in the proof of Lemma \ref{lem:existence2}; we repeat it in the following table.

\begin{table}[htbp]
\centerline{
\begin{tabular}{l|c c }
  & vn $\ell$& vn $r$\\[-0.5em]
  & $c_\ell = 0$ & $c_r = 0$\\\hline
  bt $\ell$: $\mu_\ell =1$ & $5$ & $3$\\
  bt $r$: $\mu_r =1$  & $3$ & $5$ \\
\end{tabular}}
  \caption{The price competition game in the proof of Lemma \ref{lem:existence2}.}\label{tab:example3}
  \end{table}

Consider the allocation $\xx$ where all buyers of type $\ell$ are assigned to vendor $\ell$ and all buyers of type $r$ are assigned to vendor $r$. By applying algorithm \texttt{CandidatePrice}, we obtain the price vector $\pp = (5, 5)$ and the corresponding subsidy vector $\s(\xx, \pp) = (1,1)$; this amount of subsidies is required so that no vendor has incentive to lower the price to $3-\epsilon$ and obtain all buyers  and a utility of $6-2\epsilon$. One can observe that the price vector $\pp' = (4.5, 4.5)$ and the corresponding subsidy vector $\s'(\xx,\pp') = (0.5, 0.5)$ enforces the optimal assignment as an equilibrium using half the amount of subsidies.
\end{example}

We now formally prove that no polynomial time algorithm can approximate the minimum amount of subsidies within a constant factor.

\begin{theorem}\label{thm:subsidies-apx}
Approximating {\sc MinSubsidies} within any constant is NP-hard.
\end{theorem}

\begin{proof}
We will use an approximation-preserving reduction from the {\sc Node Cover} problem in $k$-uniform hypergraphs (i.e., hypergraphs in which every edge consists of $k\geq 2$ nodes), which is formally described as follows.
\begin{quote}
{\sc Node Cover}: Given a $k$-uniform hypergraph $G$, compute a node subset $C$ of minimum size so that every (hyper)edge $e$ has at least one of its nodes in $C$.
\end{quote}
The quantity $k$ in the definition of {\sc Node Cover} is a constant. It is known that, for every constant $\epsilon>0$, approximating {\sc Node Cover} within $k-1-\epsilon$ is NP-hard \cite{DGKR05}. Given a $k$-uniform hypergraph $G$, we construct the following price competition game; see also Table~\ref{tab:sub-hardness-instance}:
\begin{itemize}
\item for every edge $e$ of $G$, there is an {\em edge} vendor $e$ and a buyer type $b_e$ with volume $1$ and valuation $(k+1)^2$ for vendor $e$;
\item for every node $j$ of $G$, there are: one {\em node} vendor $j$, one {\em auxiliary} vendor $j^*$, one buyer type $b_j$ with volume $1$ and valuations $k+3$ for vendor $j$ and $k+2$ for every vendor $e$ such that $j \in e$, and a buyer type $b_j^*$ with volume $1/(k+2)$ and valuations $k+2$ for vendor $j$ and $k+3$ for vendor $j^*$.
\item all valuations not mentioned above as well as all production costs are zero.
\end{itemize}

\begin{table}[htbp]
\centerline{
\begin{tabular}{l|c c c}
  & vn $e$& vn $j$& vn $j^*$\\[-0.5em]
  & $c_e = 0$ & $c_j=0$ & $c_{j^*}=0$ \\\hline
  bt $b_e$: $\mu_{b_e} = 1$ & $(k+1)^2$ & 0 & 0 \\
  bt $b_j$: $\mu_{b_j} = 1$ & $k+2$  & $k+3$ & 0 \\
  bt $b_j^*$: $\mu_{b_j^*} = 1/(k+2)$ & 0 & $k+2$ & $k+3$ \\
\end{tabular}}
  \caption{Part of the instance in the reduction from {\sc Node Cover}. For each edge $e$, there exists an edge vendor $e$ and a buyer type $b_e$, while for each node $j$, there are two buyer types, $b_j$ and $b_j^*$, and two vendors, $j$ and $j^*$.}
\label{tab:sub-hardness-instance}
\end{table}

In the optimal assignment $\hat\xx$, for every node $j$ of $G$, buyers of type $b_j$ are assigned to vendor $j$, and buyers of type $b_j^*$ are assigned to vendor $j^*$ and, for every edge $e$ of $G$, buyers of type $b_e$ are assigned to vendor $e$. We will show that the minimum amount of subsidies required to enforce this optimal assignment as an equilibrium is {\em equal} to the size of a minimum node cover of $G$. We will need the following lemma.

\begin{lemma}\label{lem:subsidies}
Let $\pp$ be a price vector to which assignment $\hat\xx$ is consistent. Then, for every auxiliary vendor $j^*$, edge vendor $e$, and node vendor $j$,
\begin{eqnarray}\label{eq:subfj}
s_{j^*}(\hat\xx,\pp) &=& \frac{1}{k+2}\min\{k+3,p_j+1\}-\frac{1}{k+2}p_{j^*},\\\label{eq:sube}
s_e(\hat\xx,\pp) &=& \max\{(k+1)^2,(k+1)\min_{j_i\in e}\{p_{j_i}-1\}\} - p_e,\\\nonumber
s_j(\hat\xx,\pp)&=& \max\left\{\frac{k+3}{k+2}\min\{p_{j^*}-1,1+\min\{k+2,\min_{e: j \in e}{p_e}\}\},\right.\\\label{eq:subj}
& &\quad\quad\left.\gamma_j\max\{p_{j^*}-1,1+\min\{k+2,\min_{e: j \in e}{p_e}\}\}\right\} - p_j,
\end{eqnarray}
where\footnote{The quantity $\one{X}$ equals $1$ when the event $X$ is true and is $0$ otherwise.} $\gamma_j = 1-\frac{k+1}{k+2}\one{p_{j^*}-1 \geq 1+\min\{k+2,\min_{e: j \in e}{p_e}\}}$.
\end{lemma}

\begin{proof}
First, observe that the buyers of type $b_j^*$ are still assigned to vendor $j^*$ when $j^*$ deviates to any price $p'_{j^*}$ that is at most $k+3$ (the valuation of $b_j^*$ for vendor $j^*$) and such that the assignment $\hat\xx$ is still consistent to the new price vector $(p'_{j^*},\pp_{-j^*})$, i.e., $k+3-p'_{j^*} \geq k+2-p_j$. Hence, $s_{j^*}(\hat\xx,\pp)$ is equal to the maximum (over all price deviations $p'_{j^*}$) utility of the vendor minus the vendor's utility at price $p_{j^*}$; this is expressed by (\ref{eq:subfj}).

A vendor $e$ corresponding to an edge $e=(j_1,..., j_k)$ of $G$ can attract the buyers of type $b_e$ by deviating to a price $p'_e$ up to $(k+1)^2$ for a utility of $p'_e$ (up to $(k+1)^2$). Also, vendor $e$ can attract the buyers of type $b_{j_i}$ by deviating to a price $p'_e$ such that the utility of buyers of type $b_{j_i}$ when assigned to vendor $e$ is not lower than their utility when assigned to vendor $j_i$, i.e., $k+2-p'_e \geq k+3-p_{j_i}$ or $p'_e\leq p_{j_i}-1$. In such a case, the vendor will also attract the buyers of type $b_e$ as well. When the deviation to a price $p'_e$ attracts the buyers of $b_e$ and node buyers corresponding to a subset $\overline{e}$ of at most $k-1$ nodes of $e$, the utility will be up to $k\min_{j_i\in \overline{e}}\{p_{j_i}-1\}$. This is never more than $k(k+2)$ since $p_{j_i}$ for $j_i\in \overline{e}$ must be up to $k+3$ (the valuation of buyers of type $b_{j_i}$ for vendor $j_i$) so that $\hat\xx$ is consistent to the price vector $\pp$. Finally, by deviating to any price $p'_e$ up to $\min_{j_i\in e}\{p_{j_i}-1\}$, the utility will be $(k+1)p'_e$, i.e., up to $(k+1)\min_{j_i\in e}\{p_{j_i}-1\}$. Again, $s_e(\hat\xx,\pp)$ is equal to the maximum (over all price deviations $p'_{e}$) utility of the vendor minus the vendor's utility at price $p_{e}$, which is expressed by (\ref{eq:sube}).

A node vendor $j$ can attract buyers of type $b_j$ by deviating to a price $p'_j$ up to $k+3$ (the valuation of buyers of type $b_j$ for vendor $j$) such that $k+3-p'_j\geq k+2-p_e$ for all edge vendors $e$ corresponding to edges incident to node $j$ in $G$. Hence, this price can be up to $1+\min\{k+2,\min_{e: j \in e}{p_e}\}$. Also, vendor $j$ can attract buyers of type $b_j^*$ (from vendor $j^*$) by deviating to a price $p'_j$ that satisfies $k+2-p'_j\geq k+3-p_{j^*}$; this price can be up to $p_{j^*}-1$. Hence, the maximum utility of vendor $j$ is achieved when deviating to one of the two prices. The volume of buyers attracted is $(k+3)/(k+2)$ when deviating to the lowest of the two values and either $1$ (the volume of buyers of type $b_j$) or $1/(k+2)$ (the volume of buyers of type $b_j^*$) when deviating to the highest among the two values. The quantity $\gamma_j$ in the statement of the lemma denotes the volume of buyers attracted by vendor $j$ in the latter case. Then, $s_j(\hat\xx,\pp)$ is given by the difference between the maximum utility at a deviation and the current vendor's utility in (\ref{eq:subj}).
\end{proof}

Given a node cover $C$ of $G$, we construct the price vector $\pp$ with $p_{j^*}=k+3$ for every auxiliary vendor $j^*$, $p_e=(k+1)^2$ for every edge vendor $e$, $p_j=k+2$ for every node vendor $j$ corresponding to a node $j\in C$, and $p_j=k+3$ for every other node vendor $j\not\in C$. It is easy to see that the optimal assignment is consistent to $\pp$. By Lemma \ref{lem:subsidies}, we have that $s_{j^*}(\hat\xx,\pp)=0$ for every auxiliary vendor $j^*$, $s_e(\hat\xx,\pp)=0$ for every edge vendor $e$, $s_j(\hat\xx,\pp)=1$ for every node vendor $j\in C$ and $s_j(\hat\xx,\pp)=0$, otherwise.  Obviously the cost of $\s(\hat\xx,\pp)$ is $|C|$.

We will now show that for any price vector $\pp$ to which the optimal assignment $\hat\xx$ is consistent, there is a node cover of $G$ of size at most the cost of $\s(\hat\xx,\pp)$. In order to show this, we will start with an arbitrary price vector $\pp$ to which $\hat\xx$ is consistent and will transform it, through a sequence of rounds that will not violate the consistency of $\hat\xx$ to prices and will not increase the cost of the corresponding subsidies, to a price vector with the following properties: all edge vendors have prices equal to $(k+1)^2$, auxiliary node vendors have prices equal to $k+3$, node vendors corresponding to a node cover $C$ of $G$ have prices equal to $k+2$ while the remaining node vendors have prices equal to $k+3$. The corresponding subsidies are $1$ in every node vendor corresponding to a node of $C$ and $0$ elsewhere. We describe these rounds in the following. An example of this process is presented in Figure \ref{fig:hardness} and Table \ref{tab:example_hardness}.

\begin{figure}[htbp]
\centering
\includegraphics[trim = 0 0 0 0, width = 1.8in]{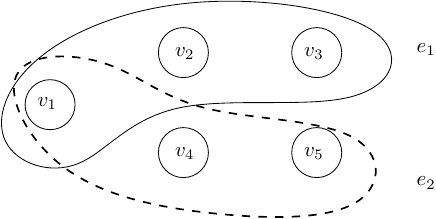}
\caption{A $3$-uniform hypergraph with five nodes and two (hyper)edges.} \label{fig:hardness}
\end{figure}

First, we consider the edge vendors one by one and increase the price of each of them to $(k+1)^2$. Clearly, the increase of the price of any edge vendor can only cause a decrease in the subsidy at the edge vendor (from equation (\ref{eq:sube})). Also, the increase of the price of any edge vendor that, initially, has value at least $k+2$ does not increase the subsidy in any auxiliary or node vendor (from equations (\ref{eq:subfj}) and (\ref{eq:subj})). Now consider a vendor corresponding to the edge $e=(j_1,..., j_k)$ that initially has a price $p_e$ of $k+2-\delta$ for some $\delta \in (0,k+2)$. Increasing this price to $(k+1)^2$ causes a decrease of $k^2+k-1+\delta$ in the subsidy of the edge vendor $e$. If $e$ does not have the minimum price among the edge vendors corresponding to edges containing node $j_i$ for every $j_i\in e$, the price increase does not affect the subsidies in node vendors $j_i \in e$. Otherwise, by inspecting the right-hand side of (\ref{eq:subj}), we have that the initial subsidies in $j_i\in e$ are at least $k+3-\delta-p_{j_i}$. Clearly, after the edge vendor price increase and since $p_j^*\leq k+3$, the subsidies in node vendor $j_i$ are at most $k+3-p_{j_i}$; the subsidies in any other vendor are not affected. Overall, we have a decrease in the total amount of subsidies by at least $k^2+k-1+\delta-k\delta$ which is strictly positive since $\delta\leq k+2$.

\begin{table}[htbp]
\centerline{
\begin{tabular}{l|c c | c c c c c | c c c c c}
  & $e_1$& $e_2$& $j_1$& $j_2$& $j_3$& $j_4$& $j_5$& $j^*_1$& $j^*_2$& $j^*_3$& $j^*_4$& $j^*_5$\\\hline
  $\pp_0$ & $5$ & $4.5$ & $5.5$ & $5.1$ & $5.1$ & $5.5$ & $4.5$ & $6$ & $5.5$ & $5$ & $6$ & $5$ \\
  $\s_0$ & $11.4$ & $11.5$ & $0.5$ & $0.9$ & $0.9$ & $0.5$ & $1$ & $0$ & $0.1$ & $0.2$ & $0$ & $0.1$\\\hline\hline
  $\pp_1$ & $16$ & $16$ & $5.5$ & $5.1$ & $5.1$ & $5.5$ & $4.5$ & $6$ & $5.5$ & $5$ & $6$ & $5$ \\
  $\s_1$ & $0.4$ & $0$ & $0.5$ & $0.9$ & $0.9$ & $0.5$ & $1.5$ & $0$ & $0.1$ & $0.2$ & $0$ & $0.1$\\\hline\hline
  $\pp_2$ & $16$ & $16$ & $5.5$ & $5.1$ & $5.1$ & $5.5$ & $5$ & $6$ & $6$ & $6$ & $6$ & $6$ \\
  $\s_2$ & $0.4$ & $0$ & $0.5$ & $0.9$ & $0.9$ & $0.5$ & $1$ & $0$ & $0$ & $0$ & $0$ & $0$\\\hline\hline
  $\pp_3$ & $16$ & $16$ & $5.5$ & $5.1$ & $5$ & $5.5$ & $5$ & $6$ & $6$ & $6$ & $6$ & $6$ \\
  $\s_3$ & $0$ & $0$ & $0.5$ & $0.9$ & $1$ & $0.5$ & $1$ & $0$ & $0$ & $0$ & $0$ & $0$\\\hline\hline
  $\pp$ & $16$ & $16$ & $6$ & $6$ & $5$ & $6$ & $5$ & $6$ & $6$ & $6$ & $6$ & $6$ \\
  $\s$ & $0$ & $0$ & $0$ & $0$ & $1$ & $0$ & $1$ & $0$ & $0$ & $0$ & $0$ & $0$
\end{tabular}}
  \caption{The process of adjusting the price and subsidy vectors for the vendors defined by the hypergraph of Figure \ref{fig:hardness}. Recall that valuations are as described in Table \ref{tab:sub-hardness-instance}, using $k=3$, i.e., $v_{b_ee} = 16$ for the edge vendor and the buyer type due to edge $e$, $v_{b_je} = 5$ when $j \in e$, $v_{b_jj}=6$, $v_{b^*_jj}=5$ and $v_{b^*_jj^*}=6$. Let $\pp_0$ be a price vector to which the optimal assignment is consistent and $\s_0$ the corresponding subsidy vector. We present how these vectors change during each round, until we reach the final vectors $\pp$ and $\s$. We remark that these final vectors correspond to selecting nodes $v_3$ and $v_5$ as the node cover of the hypergraph of Figure \ref{fig:hardness}.}\label{tab:example_hardness}
  \end{table}

After the above round, each edge vendor has a price of $(k+1)^2$ and, since $p_{j^*}\leq k+3$ for any vendor $j^*$, the subsidy of node vendor $j$ is exactly $k+3-p_j$. In a second round, we consider all node vendors with prices strictly smaller than $k+2$ and increase them to $k+2$ and all auxiliary vendors with prices strictly smaller than $k+3$ and increase them to $k+3$. Observe that (from (\ref{eq:subfj})) the subsidy of an auxiliary vendor is $0$ after this round (i.e., it does not increase), the subsidies at edge vendors are not affected (since no node vendor price increases above $k+2$; see (\ref{eq:sube})), and the subsidies in node vendors (which are now given by the expression $k+3-p_j$) can only decrease.

So, after the first two rounds, all auxiliary vendor prices are equal to $k+3$, all edge vendor prices are equal to $(k+1)^2$, and the node vendor prices are in $[k+2,k+3]$. In the third round, we consider all edges of $G$ one by one; for each edge $e=(j_1,..., j_k)$, we set the minimum price among $j_1$, $j_2$, ..., $j_k$ to $k+2$. Let $j'$ be this vendor; assuming that the price of $j'$ is $k+2+\delta$ when it is considered (with $\delta\in [0,1]$), this causes a decrease of $(k+1)\delta$ in the subsidy of the edge vendor $e$ (as well as possible decreases in the subsidies of other edge vendors corresponding to other edges containing node $j'$) and an increase of $\delta$ in the subsidy of the node vendor $j'$, i.e., no overall increase. During the last round, all node vendor prices that are strictly higher than $k+2$ are set to $k+3$. Increasing the price of a node vendor $j$ can only decrease the subsidies in $j$, does not affect the subsidies in $j^*$, and does not affect the subsidy of any edge vendor corresponding to an edge $e$ containing node $j$, since there is a node $j'\in e$ so that the price of the node vendor $j'$ has already been set to $k+2$, implying an edge vendor subsidy of zero (this follows from (\ref{eq:sube})). Again, this last round does not increase the total amount of subsidies. Observe that the nodes corresponding to node vendors with price equal to $k+2$ form a node cover in $G$. The subsidies in each of these node vendors is exactly $1$ while no subsidies correspond to the remaining vendors.
\end{proof}

\section{Discussion and open problems}\label{sec:open}
In this work, we have posed and answered a long list of questions about price competition games among single-product vendors. Of course, our work reveals a lot more open problems. We mention a few here. First, we remark that identifying broad classes of games that always admit equilibria is very interesting. It is not hard to see that games with {\em single-minded} buyers, i.e., where each buyer only wishes to be assigned to a single vendor, always have pure equilibria, as essentially each vendor sets the price that maximizes its utility and the buyers either abstain or choose their preferred vendor; note that, in this setting, there is no price competition among vendors. Another class of games that always admits pure equilibria is the following extension of the case in Lemma \ref{lem:existence}: let there be $n$ buyer types, where $v_{ij} = v_{1j}+c_i$, for $2\leq i\leq n$ where $c_i\geq 0$ is a constant. Note that, in this setting, all buyer types agree on the same ranking of vendors and that, for fixed $i$, $v_{ij}-v_{ij'}$ is constant for any pair of vendors $j,j'$; we can adapt the proof of Lemma \ref{lem:existence} to prove this claim. However, one can show that the case, where all buyer types agree on the same ranking of vendors but $v_{ij} - v_{ij'}$ may vary for different pairs $j$, $j'$, does not always admit equilibria. To see that, consider a game with two vendors of production cost $0$ and two buyer types of volume $1$. Let $v_{11} = 6$, $v_{12} = 4$, $v_{21} = 5$ and $v_{22}=0$, where again $v_{ij}$ denotes the valuation of a buyer of type $i$ for vendor $j$; it is not hard to show that this game does not admit equilibria.

From the algorithmic point of view, observe that we have made no particular attempt to optimize the running time of our algorithms. We believe that there is much room for improvement on the running time of \texttt{CandidatePrice} and \texttt{Enumerate}. In particular, it would be interesting to come up with FPT algorithms (see Downey and Fellows~\cite{DF99}) for {\sc PriceCompetition} with respect to different parameters. Second, in spite of our inapproximability result (Theorem \ref{thm:subsidies-apx}), we believe that it is important to design polynomial-time approximation algorithms for {\sc MinSubsidies}. For example, is there a logarithmic approximation algorithm? What about additive approximations using an amount of subsidies that exceeds the minimum by at most $\rho \cdot\SW^*$ for some small $\rho>0$?

Another set of open problems comes from introducing constraints to price competition games such as supply limitations. For example, consider additional input parameters that indicate the maximum volume of buyers each vendor can support. We believe that this subtle difference in the definition makes the setting even richer from the computational point of view. Another question concerns mixed equilibria. Do such equilibria always exist? Observe that the strategy spaces of vendors have infinite size in this case. Can they be computed efficiently? What is their price of anarchy? What about generalizations of our model that include uncertainty for buyer valuations? It is our firm belief that these questions deserve investigation.

An alternative model for using subsidies is to reimburse the vendors so that a given buyers-to-vendors assignment $\xx$ results in an equilibrium with the vendors figuring out the corresponding price vector $\pp$ on their own. Instead, our model reimburses the vendors so that a given pair $(\xx,\pp)$ is an equilibrium. Apart from the requirement that the vendors compute collectively an appropriate price vector, we remark that there are cases where this variant leads to higher subsidies; the price competition game in Example \ref{ex:lemma4} is such a case.

Finally, it makes sense to consider vendors with concave production costs (that model economies of scale), as opposed to fixed production costs per unit that we consider in the paper. Clearly, all our negative results, such as the existence of price competition games with no pure Nash equilibria and the NP-completeness results, still hold but we do not know if the positive results carry over to this more general setting. We remark that this modification also affects the consequences after a price change. For example, a vendor may decrease its price in order to attract more buyers and, after these buyers select that vendor, the vendor can further reduce its price as now its production cost per item has decreased. These model variants deserve investigation as well.

\bibliography{VendorsBuyers}
\end{document}